\let\accentvec\vec 
\documentclass[runningheads, envcountsame, a4paper]{llncs} 
 
\let\vec\accentvec

\usepackage[utf8]{inputenc}

\title{Non-Abelian Gauge-Invariant Cellular Automata}
\titlerunning{Non-abelian Gauge-Invariant CA}

\author{Pablo Arrighi\inst{1,2}\orcidID{0000-0002-3535-1009} \and Giuseppe Di Molfetta\inst{1,3}\orcidID{0000-0002-1261-7476} \and Nathana\"el Eon\inst{1,4}\orcidID{0000-0003-0649-0891}}
\authorrunning{P. Arrighi, G. Di Molfetta, N. Eon}

\tocauthor{Pablo Arrighi, Giuseppe Di Molfetta, Nathana\"el EON}
\toctitle{Non-Abelian Gauge-Invariant Cellular Automata}

\institute{Aix-Marseille Univ, Universit\'e de Toulon, CNRS, LIS, Marseille, France
\and
IXXI, Lyon, France
\and
Departamento de F{\'{i}}sica Te{ó}rica and IFIC, Universidad de Valencia-CSIC, Dr. Moliner 50, 46100-Burjassot, Spain
\and 
\'Ecole Centrale, France
}

\usepackage{amsmath, amsfonts, amssymb}
\usepackage{color}

\usepackage{tikz}
\usepackage{ifthen}

\usepackage{xcolor}

\usepackage{mathtools}  
\mathtoolsset{showonlyrefs}

\newcommand{\gamm}{\bar{\gamma}}
\newcommand{\Gamm}{\bar{\Gamma}}

\newcommand{\ZZ}{\mathbb{Z}}

\begin{document}

\maketitle

\begin{abstract}
Gauge-invariance is a mathematical concept that has profound implications in Physics---as it provides the justification of the fundamental interactions. It was recently adapted to the Cellular Automaton (CA) framework, in a restricted case. In this paper, this treatment is generalized to non-abelian gauge-invariance, including the notions of gauge-equivalent theories and gauge-invariants of configurations.
\keywords{cellular automata \and gauge-invariance\and quantum information.}
\end{abstract}

\section{Introduction}

In Physics, symmetries are essential concepts used to derive the laws which model nature. Among them, gauge symmetries are central, since they provide the mathematical justification for all four fundamental interactions: the weak and strong forces (short range interactions), electromagnetism \cite{quigg2013gauge} and to some extent gravity (long range interactions).
In Computer Science, cellular automata (CA) constitute the most established model of computation that accounts for euclidean space. Yet its origins lies in Physics, where they were first used to model hydrodynamics and multi-body dynamics, and are now commonly used to model particles or waves. The study of gauge symmetries in CA is expected to benefit both fields. In order to obtain discrete systems that can simulate physics on the one hand. In order to bring gauge theory to Computer Science, as a tool to study redundancy and fault-tolerant computation for instance, on the other hand.

A study of gauge symmetries in CA has been recently studied, by the same authors, in the particular case of abelian gauge symmetries \cite{arrighi2018gauge}. Here, we provide a generalization to non-abelian gauge symmetries. In Physics, the generalization from abelian to non-abelian gauge theories was a non-trivial but crucial step, that enabled taking into account a wider range of phenomena.

The paper is organized as follows. Sec. \ref{sec:gaugeinv} is a reformulation in a more general framework of the gauge-invariance in CA definitions and procedure given in \cite{arrighi2018gauge}. It provides the context and notations used in the rest of the paper. In Sec. \ref{sec:nonabelian}, a complete example of non-abelian gauge-invariant CA is given through the application of the \textit{gauging procedure}. It provides an example of the route one may take in order to obtain a gauge-invariant CA, starting from one that does not implement the symmetry. Sec. \ref{sec:equivalence} discusses the equivalence of theories and develops the notion of invariant sets. We summarize in Sec. \ref{sec:conclusion} and provide related works perspectives.

\section{Gauge-invariance}

\paragraph{Theory to be gauged.\label{sec:gaugeinv}} 
In this paper, {\em theories} stands for CA. We start from a theory $R$, which internal state space is $\Sigma$ and local rule is $\lambda_R$. We denote by $\psi_{x,t}$ the state of the cell at position $x$ and time $t$. $\psi_t$ denotes a configuration which is a function from $\ZZ$ into $\Sigma$ that gives a state for each position $x$. As a running example, we pick possibly the simplest and most nature physics-like reversible CA (RCA) : one that has particles moving left and right. More precisely, in this example $\Sigma = \{ 0, ..., N\} ^2$, therefore, we can write $\psi_{x,t}=(\psi^l_{x,t}, \psi^r_{x,t})$ where the exponents $l$ and $r$ denotes the left and right parts, each being an element of $\{ 0, ..., N\}$. The local rule $\lambda_R$ takes the right-incoming left sub-cell (i.e. $\psi^l_{x+1}$) to the left, and the left-incoming right sub-cell (i.e. $\psi^r_{x-1}$) to the right:
\begin{align}\label{eq:localrule}
\psi_{x,t+1}=\left(\psi^l_{x,t+1},\psi^r_{x,t+1}\right)=\lambda_R\left(\psi^r_{x-1,t},\psi^l_{x+1,t}\right)=\left(\psi^l_{x+1,t},\psi^r_{x-1,t}\right)
\end{align}
Such a CA is said to be expressed in the block-circuit form which is often referred as the (Margolus-)Partitioned CA in Computer Science vocabulary \cite{ToffoliMargolusModelling}, or Lattice-gas automata in Physics \cite{wolf2004lattice}. Fig-\ref{fig:1example} gives an example of this dynamics for $N=2$ (where the three colors represent the three possible states), and Fig-\ref{fig:1framework} introduces the conventions used.

\begin{figure}[ht]
\begin{minipage}[t]{0.45\textwidth}
    \Large
    \centering
    \resizebox{\textwidth}{!}{\begin{tikzpicture}


\draw[color=gray] (5.5,0) -- (12.5,7);
\draw[color=gray] (2.5,1) -- (8.5,7);
\draw[color=gray] (10,0.5) -- (14,4.5);
\draw[color=gray] (2,4.5) -- (4.5,7);

\draw[color=gray] (4,0.5) -- (2.5,2);
\draw[color=gray] (14,2.5) -- (9.5,7);
\draw[color=gray] (8,0.5) -- (2.5,6);
\draw[color=gray] (12,0.5) -- (5.5,7);

\node at (3, 1.5) [draw,scale=1,circle,color=gray, fill=white]{}; 
\node at (7, 1.5) [draw,scale=1,circle,color=gray, fill=white]{}; 
\node at (11, 1.5) [draw,scale=1,circle,color=gray, fill=white]{}; 

\node at (5, 3.5) [draw,scale=1,circle,color=gray, fill=white]{}; 
\node at (9, 3.5) [draw,scale=1,circle,color=gray, fill=white]{}; 
\node at (13, 3.5) [draw,scale=1,circle,color=gray, fill=white]{}; 

\node at (3, 5.5) [draw,scale=1,circle,color=gray, fill=white]{}; 
\node at (7, 5.5) [draw,scale=1,circle,color=gray, fill=white]{}; 
\node at (11, 5.5) [draw,scale=1,circle,color=gray, fill=white]{};

  \filldraw[color=black, fill=white, thick](4,0) rectangle (5,1);
  \filldraw[color=black, fill=black, thick](5,0) rectangle (6,1);
  \filldraw[color=black, fill=white, thick](8,0) rectangle (9,1);
  \filldraw[color=black, fill=white, thick](9,0) rectangle (10,1);
  \filldraw[color=black, fill=gray, thick](12,0) rectangle (13,1);
  \filldraw[color=black, fill=white, thick](13,0) rectangle (14,1);

  \filldraw[color=black, fill=white, thick](2,2) rectangle (3,3);
  \filldraw[color=black, fill=white, thick](3,2) rectangle (4,3);
  \filldraw[color=black, fill=white, thick](6,2) rectangle (7,3);
  \filldraw[color=black, fill=black, thick](7,2) rectangle (8,3);
  \filldraw[color=black, fill=gray, thick](10,2) rectangle (11,3);
  \filldraw[color=black, fill=white, thick](11,2) rectangle (12,3);
    
  \filldraw[color=black, fill=white, thick](4,4) rectangle (5,5);
  \filldraw[color=black, fill=white, thick](5,4) rectangle (6,5);
  \filldraw[color=black, fill=gray, thick](8,4) rectangle (9,5);
  \filldraw[color=black, fill=black, thick](9,4) rectangle (10,5);
  \filldraw[color=black, fill=white, thick](12,4) rectangle (13,5);
  \filldraw[color=black, fill=white, thick](13,4) rectangle (14,5);

  \filldraw[color=black, fill=white, thick](2,6) rectangle (3,7);
  \filldraw[color=black, fill=white, thick](3,6) rectangle (4,7);
  \filldraw[color=black, fill=gray, thick](6,6) rectangle (7,7);
  \filldraw[color=black, fill=white, thick](7,6) rectangle (8,7);
  \filldraw[color=black, fill=white, thick](10,6) rectangle (11,7);
  \filldraw[color=black, fill=black, thick](11,6) rectangle (12,7);


\draw[color=black] (1, 0.5) node {t};
\draw[color=black] (1, 2.5) node {t+1};
\draw[color=black] (1, 4.5) node {t+2};
\draw[color=black] (1, 6.5) node {t+3};

\draw[color=black] (3, -1) node {x-2};
\draw[color=black] (5, -1) node {x-1};
\draw[color=black] (7, -1) node {x};
\draw[color=black] (9, -1) node {x+1};
\draw[color=black] (11, -1) node {x+2};
\draw[color=black] (13, -1) node {x+3};

\end{tikzpicture}}
    \caption{\label{fig:1example}A spacetime diagram of $R$.}
\end{minipage}\hfill
\begin{minipage}[t]{0.45\textwidth}
    \centering
    \resizebox{\textwidth}{!}{\begin{tikzpicture}

\draw[thick] (2,1) -- (4,2)
    (6,1) -- (4,2)
    (3.8,2) -- (3.8, 3)
    (4.2,2) -- (4.2, 3);

\node at (4, 2) [draw,scale=2,circle,color=black, fill=white,thick]{$\lambda_R$}; 

\filldraw[color=black, dashed, fill=white, thick](0,0) rectangle (1,1);
\filldraw[color=black, fill=white, thick](1,0) rectangle (2,1);

\filldraw[color=black, fill=white, thick](6,0) rectangle (7,1);
\filldraw[color=black, dashed, fill=white, thick](7,0) rectangle (8,1);

\filldraw[color=black, fill=white, thick](3,3) rectangle (4,4);
\filldraw[color=black, fill=white, thick](4,3) rectangle (5,4);

\draw (3.5,3.5) node {\color{gray} $\psi^l$};
\draw (4.5,3.5) node {\color{gray} $\psi^r$};

\draw (1, -0.5) node {\large x-1};
\draw (4, -0.5) node {\large x};
\draw (7, -0.5) node {\large x+1};

\draw (-0.5, 0.5) node {\large t};
\draw (-0.5, 3.5) node {\large t+1};

\end{tikzpicture}}
    \caption{\label{fig:1framework}Conventions.}
\end{minipage}
\end{figure}

This theory $R$ is {\em to be gauged} because does not yet implement the gauge symmetry, which is a local invariance under a group of operators called {\em gauge transformations.} The theory $R$ will eventually be extended into a theory $T$ that does implement the symmetry, through a {\em gauging procedure}.

\paragraph{Gauge transformations.} The gauge symmetry is an invariance of the evolution under a {\em gauge transformation}. What we call a gauge transformation is based on a monoid 
$\Gamma$ of operators acting on the internal state space $\Sigma$. A gauge transformation is the application, onto the state of each cell in a configuration, of one of the operators of $\Gamma$. More formally, for each cell $x$ is attributed an element $\gamma_x$ in $\Gamma$. Thereby specifying a gauge transformation $\gamm$ acting over entire configurations:
\begin{equation}
    \gamm :     \begin{matrix} 
        \Sigma^\ZZ & \rightarrow & \Sigma^\ZZ \\
        c  & \mapsto & \big(x\mapsto \gamma_x (c_x)\big)
    \end{matrix}
\end{equation}
We denote $\Gamm \cong \Gamma^\ZZ$ the set of these gauge transformations $\gamm$.

\paragraph{Gauge-invariance.} A theory $T$ is said gauge-invariant if its evolution is impervious to gauge-transformations. In other words, applying a gauge transformation $\gamm$ followed by the evolution $T$, \textit{amounts to the same} as applying the evolution $R$ directly. What is meant by \textit{amounts to the same} is that both outputs are the same, up to a gauge-transformation $\gamm'$. 
Given that we want the evolution $T$ to be deterministic, we impose that $\gamm'$ be determined from $\gamm$ by means of some theory $Z$. 
After those consideration, gauge-invariance can be defined as follows which is a reformulation from \cite{arrighi2018gauge} :
\begin{definition}[Gauge-invariance]
A theory $T$ is gauge-invariant if and only if there exists $Z$ a theory such that for all $\gamm \in \Gamm$
\begin{align}\label{eq:gaugeinvariance}
Z(\gamm) \circ T=T\circ \gamm
\end{align}
where the symbol $\circ$ represents the composition.
\end{definition}

The gauge-invariance is represented in Fig-\ref{fig:1gaugetransform} where $\gamma_0, \gamma_1$ and $\gamma'$ are local gauge transformations.
\begin{figure}[ht]
    \Large
    \centering
    \resizebox{0.7\textwidth}{!}{\newcommand{\stateTikz}[6]{
  \ifthenelse{#5 > 0 \AND #6>0}{\def\mycola{white}}{\def\mycola{black}};
  \ifthenelse{#5 > 0}{\def\mycolb{black}}{\def\mycolb{white}};
  \ifthenelse{#6 > 0}{\def\mycolc{black}}{\def\mycolc{white}};
  \filldraw[color=\mycola, fill=\mycolb, thick](#1- #3 , #2) rectangle (#1 , #2 + #4);
  \filldraw[color=\mycola, fill=\mycolc, thick](#1, #2) rectangle (#1 + #3 , #2 + #4);
}

\newcommand{\basegraph}[4]{
	\stateTikz{0}{0}{1}{1}{#1}{#1};
	\stateTikz{4}{0}{1}{1}{#2}{#2};
	\stateTikz{2}{3}{1}{1}{#3}{#4};
	\draw (0.5,1) -- (2.5,3)
		(3.5,1) -- (1.5,3)
		(-0.5,1) -- (-1.5, 2)
		(4.5,1) -- (5.5,2);
	\draw (1.5,4) -- (0.5,5)
		(2.5,4) -- (3.5,5);
}

\begin{tikzpicture}
	
	\basegraph{0}{1}{0}{1};
	\draw[color=black, fill=white, thick] (0,1.6) ellipse (1.2 and 0.4) node {$\gamma_0$};
	\draw[color=black, fill=white, thick] (4,1.6) ellipse (1.2 and 0.4) node {$\gamma_1$};

	\draw[very thick, dashed] (6.5,0) -- (6.5,6);
	\fill [white] (6,2.5) rectangle (7,3.5);
	\draw (6.5, 3) node {=};
	
	\begin{scope}[shift={(9,0)}]
		\basegraph{0}{1}{1}{0};
		\draw[color=black, fill=white, thick] (2,4.6) ellipse (1.2 and 0.4) node {$\gamma'$};
	\end{scope}
 
\end{tikzpicture}}
    \caption{\label{fig:1gaugetransform}Illustration of gauge-invariance.}
\end{figure}
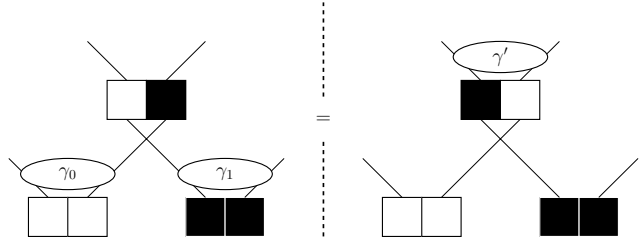

\paragraph{Gauging procedure.} In order to extend the non-gauge-invariant theory $R$ into a gauge-invariant theory $T$ we will apply a gauging procedure, which is strongly inspired from Physics. The procedure begins by introducing new information, namely the gauge field $A$, at each point in spacetime, and to extend the theory $R$ into an $A$-dependant theory $R_A$ that features gauge-invariance. We also need to decide how that gauge field changes under gauge transformations. In order to keep the notations simple, we will write $\gamm(A)$ for the gauge transformation of the gauge field $A$. Even though the gauge field may not transform the same way as the initial configuration, the context shall be enough to lift any ambiguity. Let us make precise what we mean by gauge-invariance, whenever a theory depends on an external field. 
\begin{definition}[Inhomogeneous gauge-invariance]
A theory $T_\bullet$ is said to be {\em inhomogeneous gauge-invariant} if and only if there exists $Z$ a theory such that for all $\gamm \in \Gamm$,
\begin{align}\label{eq:inhomoggaugeinvariance}
Z(\gamm) \circ T_\bullet=T_{\gamm (\bullet)}\circ \gamm
\end{align}
\end{definition}
Typically, this condition puts a strong constraint on the way the gauge transformation must act over the gauge field, i.e. $\gamm(A)$. An example is given in Sec.\ref{sec:procedure3}.

Having determined such gauge transformation, the final step of the gauging procedure is to give a theory that specifies the dynamics of the gauge field. This theory, joint together with $R_A$, must give a global theory $T$ that verifies the original gauge-invariance condition \eqref{eq:gaugeinvariance}. Again, an example is given in Sec.\ref{sec:procedure4}.

All-in-all, the gauging procedure can be summarized in these four steps which we will use as a basis for the rest of the paper :
\begin{enumerate}
    \item Start with a theory to be gauged $R$ and a set of gauge transformations $\Gamm$.
    \item Introduce the gauge field $A$, transform $R$ into $R_A$.
    \item Define $\gamm(A)$ through the requirement that $R_A$ verifies condition \eqref{eq:inhomoggaugeinvariance}.
    \item Give a theory $A$ in order to define a global gauge-invariant theory $T$.
\end{enumerate}
The first two steps are free for the user/physicist to choose according to the system to be modelled. The third step however is mostly determined by the gauge-invariance condition. Finally, the degree of freedom in the choice of the dynamics for $A$ -- i.e. the last step -- may depend on the specific cases and no general characterization of the leftover degrees of freedom exists. In the abelian case however, the {\em gauge fixing soundness} result helps \cite{arrighi2018gauge}.

\section{Non-abelian gauge-invariance\label{sec:nonabelian}}

In gauge theories, we use the term abelian or non-abelian to refer to the (non-) commutativity of the monoid $\Gamma$ of operations over $\Sigma$, or equivalently to that of the gauge-transformations $\Gamm$. In physics, abelian gauge theories give rise to electromagnetism, while non-abelian gauge-theories (Yang-Mills theories in particular) allow for the formulation of the whole standard model---namely the electromagnetic, weak and strong interactions. Whether gravitation is a non-abelian gauge theory is open to interpretation, but it certainly has some flavour of it. By-the-way, non-abelian really means possibly-non-abelian, it still comprises the abelian subcase.
In this section, we produce a complete example of a non-abelian gauge-invariant CA by applying the gauging procedure. 

\subsection*{Back to the running example.}
{\em Step 1.} Recall that our point of departure was the theory $R$, whose rule $\lambda_R$ given by Eq.\eqref{eq:localrule} is takes two subcells into two subcells through a bijection.  
In order to have as simple an example as possible, we choose $N=2$, thus $\Sigma=\{0,1,2\}^2$. Let us choose a monoid $\Gamma$ that follows the same structure as in \cite{arrighi2018gauge}, i.e. so that the operators $\gamma \in \Gamma$ act identically on both elements of $\Sigma$ (this choice is traditional in gauge theories, but not a necessity of our definitions). That is, they act by applying the same permutation on both subcells. More formally, let us denote by $S(N)$ the set of permutations over $\{0,...,N-1\}$, we let:
\begin{equation}
\Gamma=\{ s \otimes s\; | \; s\in S(N) \}.
\end{equation}
Given some $\gamma=s \otimes s$, the notations $\gamma^l=\gamma^r$ will be short for $s$.~\\[5mm]  
{\em Step 2.} This step is to introduce an external gauge field $A$ and make $R$ into an $A$-dependent rule $R_A$. We take the gauge field $A$ to be defined at every half-integer space position (and every time step). This definition is physics-inspired and corresponds to the convention used in \cite{arrighi2019quantum}. We let $A$ take its values in $S(3)$ the set of permutations over 3 elements. We let $R_A$ be defined by the $A$-dependent local rule $\lambda_{R_A}$, which is spacetime-dependent since $A$ is spacetime-dependent: 
\begin{equation}
\left(\lambda_{R_A}\right)_{x,t}= \lambda_{R} \circ (A_{x-1/2, t}\otimes A_{x+1/2,t}^{-1})
\end{equation}
The induced evolution that used to be described by Eq.\eqref{eq:localrule} now becomes :
\begin{align}\label{eq:localruleA}
    \left(\psi^l_{x,t+1},\psi^r_{x,t+1}\right) &= \left(\lambda_{R_A}\right)_{x,t} \left(\psi^r_{x-1,t}, \psi^l_{x+1,t}\right) = \left(A_{x+1/2,t}^{-1}\psi^l_{x+1,t},A_{x-1/2,t}\psi^r_{x-1,t}\right)
\end{align}
The local rule $\lambda_{R_A}$ is represented in Fig-\ref{fig:2frame}.

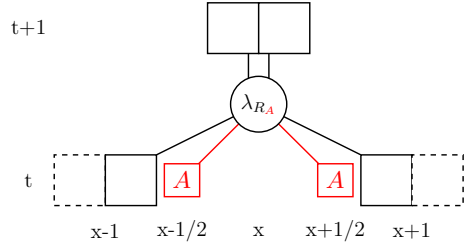
\begin{figure}
    \centering
    \resizebox{0.5\textwidth}{!}{\begin{tikzpicture}

\draw[thick] (2,1) -- (4,2)
    (6,1) -- (4,2)
    (3.8,2) -- (3.8, 3)
    (4.2,2) -- (4.2, 3);
\draw[thick, red] (2.5,0.5) -- (4,2) -- (5.5,0.5);

\node at (4, 2) [draw,scale=1,circle,color=black, fill=white,thick]{\large $\lambda_{R \color{red}_A}$};

\filldraw[color=black, dashed, fill=white, thick](0,0) rectangle (1,1);
\filldraw[color=black, fill=white, thick](1,0) rectangle (2,1);

\filldraw[color=black, fill=white, thick](6,0) rectangle (7,1);
\filldraw[color=black, dashed, fill=white, thick](7,0) rectangle (8,1);

\filldraw[color=black, fill=white, thick](3,3) rectangle (4,4);
\filldraw[color=black, fill=white, thick](4,3) rectangle (5,4);

\node at (2.5, 0.5) [draw,scale=1.5,color=red, fill=white,thick]{$A$}; 
\node at (5.5, 0.5) [draw,scale=1.5,color=red, fill=white,thick]{$A$};

\draw (1, -0.5) node {\large x-1};
\draw (4, -0.5) node {\large x};
\draw (7, -0.5) node {\large x+1};
\draw (2.5, -0.5) node {\large x-1/2};
\draw (5.5, -0.5) node {\large x+1/2};

\draw (-0.5, 0.5) node {\large t};
\draw (-0.5, 3.5) node {\large t+1};

\end{tikzpicture}}
    \caption{Introducing the gauge field.\label{fig:2frame}}
\end{figure}

The way $R_A$ depends on the gauge field and the definition of the gauge field itself is motivated through the fact that $A$ can be made to cancel any gauge-transformation done on the input.~\\[5mm]
{\em Step 3.\label{sec:procedure3}} The gauge transformation of the gauge field $A$ is dictated by the condition \eqref{eq:gaugeinvariance}. Such condition can be developed locally due to the locality of the theory and of the gauge transformation. It gives : there exists $Z$ a theory such that for all $\gamm \in \Gamm$ and $x\in \ZZ$
\begin{align}
    Z(\gamm)_{x} \circ (\lambda_{R_A})_x = (\lambda_{R_{\gamm(A)}})_{x} \circ (\gamma_{x-1}^r \otimes \gamma_{x+1}^l).
\end{align}
Replacing the local rule by its expression gives the following equation:
\begin{align}
    Z(\gamm)_{x} &\circ \lambda_R \circ (A_{x-1/2} \otimes A_{x+1/2}^{-1}) \\ 
     &= \lambda_R  \circ (\gamm(A)_{x-1/2}\otimes \gamm(A)_{x+1/2}^{-1}) \circ (\gamma_{x-1}^r \otimes \gamma_{x+1}^l).
\end{align}

This equation is equivalent to the following system
\begin{align}
    &\begin{cases}
        \gamm(A)_{x-1/2} \circ \gamma_{x-1}^l = Z(\gamm)_x^l \circ A_{x-1/2}   \\
        \gamm(A)_{x+1/2}^{-1} \circ \gamma_{x+1}^r = Z(\gamm)_x^r \circ A_{x+1/2}^{-1}  
    \end{cases} \\
    &\Leftrightarrow \begin{cases}
        \gamm(A)_{x-1/2}  = Z(\gamm)_x^l \circ A_{x-1/2} \circ (\gamma^l_{x-1})^{-1}  \\
        \gamm(A)_{x+1/2} =  \gamma_{x+1}^r \circ A_{x+1/2}  \circ (Z(\gamm)_x^r)^{-1}
    \end{cases}
\end{align}

Such a system gives an information and a constraint. First, the gauge transformation of the gauge field $A$ is given explicitly in terms of that over $\psi$, which was the main objective. Second, it puts some constraints over $Z$. In order to satisfy both these equations for any $A$ and $\gamm$ given as input, the choice of $Z$ is limited. For instance, $Z$ cannot be a translation to the right, because that would impose for $\gamma$ to be the same at every position. One solution is to choose $Z(\gamm) = \gamm$. Such choice, common in physics, was also taken in \cite{arrighi2019quantum} which gives a quantum CA for one-dimensional QED (quantum electrodynamics).

In the end, the gauge-transformation of the gauge field $A$ reads, for $x$ an half-integer and using $\gamma^r=\gamma^l$:
\begin{equation} \label{eq:gaugetransfA}
    \gamm(A)_{x}  = \gamma_{x+1/2}^l \circ A_{x} \circ (\gamma_{x-1/2}^l)^{-1} 
\end{equation}

{\em Step 4.\label{sec:procedure4}} We now have an inhomogeneous gauge-invariant theory $R_A$, with respect to $\Gamma$ and $Z=I$. The last step is to provide a theory for the dynamics of the gauge field $A$, in order to yield a complete gauge-invariant theory $T$ that evolves both $\psi$ and $A$---i.e. over the internal state space $S(N) \times \Sigma \times S(N)$. The usual way to do this is to propose an inhomogeneous gauge-invariant theory $S_\psi$, with respect to the same $\Gamma$ and $Z$, that acts on $A$ but depend on $\psi$. Then combining $R_A$ and $S_\psi$, which are both inhomogeneous gauge-invariant, will give a gauge-invariant theory $T$ with respect to $\Gamma$ and $Z$. Let us justify this by writing down the inhomogeneous gauge-invariance condition \eqref{eq:inhomoggaugeinvariance} for $R_A$ and $S_\psi$: for all $\psi\in\Sigma^\ZZ$,  $A\in S(N)$ and $\gamm\in \Gamm$:
\begin{align}
    R_{\gamm(A)} \circ \gamm (\psi) &= Z(\gamm) \circ R_A (\psi) \\
    S_{\gamm(\psi)} \circ \gamm (A) &= Z(\gamm) \circ S_\psi (A) 
\end{align}
Combining both, we obtain
\begin{align}\label{eq:proofofT}
    T\circ \gamm (\psi, A) 
     &= \Big( R_{\gamm(A)} \circ \gamm (\psi), \: S_{\gamm(\psi)} \circ \gamm (A)  \Big) \\
     &= \Big( Z(\gamm) \circ R_A (\psi), \: Z(\gamm) \circ S_\psi (A) \Big) \\
     &= Z(\gamm) \circ T (\psi, A)
\end{align}
which is exactly the gauge-invariance condition \eqref{eq:gaugeinvariance} for $T$.

In order to find a suitable $S_\psi$, one therefore has to write the inhomogeneous gauge-invariance condition \eqref{eq:inhomoggaugeinvariance}, substituting for $Z$ and $\gamma$ by their definitions, including $\gamma(A)$ as given by \eqref{eq:gaugetransfA}. Several possible $S_\psi$ may meet this condition: to the our best knowledge there is no general notion of a minimal $S_\psi$. However, the running example does exhibit a minimal solution which is the identity. The inhomogeneous gauge-invariance can then be verified easily : for all $x\in \ZZ$ and $\gamm \in \Gamm$,
\begin{align}
    Z(\gamm)_x \circ (\lambda_{S_\psi})_x &= (\lambda_{S_{\gamm(\psi)}})_x \circ \gamma_x 
    \Longleftrightarrow \gamma_x \circ I = I \circ \gamma_x.
\end{align}
Combining $R_A$ with $S_\psi =I$, gives a gauge-invariant theory $T$ with local rule $\lambda_T$ as follows: for any spacetime position $x,t$,
\begin{align}
     \left(A_{x-1/2}, \psi^l_{x}, \psi^r_{x},  A_{x+1/2}\right)_{t+1}  &= \lambda_T \left(\psi^r_{x-1}, A_{x-1/2}, A_{x+1/2}, \psi^l_{x+1}\right)_t \\ 
     &= \left(A_{x-1/2},A_{x+1/2}^{-1}\psi^l_{x+1},A_{x-1/2}\psi^r_{x-1},A_{x+1/2} \right)_t
\end{align}
where the final time index applies to every element of the list. This rule is illustrated in Fig-\ref{fig:2full}. 

\begin{figure}[ht!]
    \centering
    \resizebox{0.45\textwidth}{!}{\begin{tikzpicture}

\draw[thick] (2,1) -- (4,2);
\draw[thick] (6,1) -- (4,2);
\draw[thick] (3.8,2) -- (3.8, 3)
    (4.2,2) -- (4.2, 3);
\draw[thick, red, dashed] (2.5,0.5) -- (2.5,3.5)
                        (5.5,3.5) -- (5.5,0.5);

\node at (2.5,1.25) [draw, scale=0.6, circle, color=red, fill=red, thick]{};
\node at (5.5,1.25) [draw, scale=0.6, circle, color=red, fill=red, thick]{};


\node at (4, 2) [draw,scale=1,circle,color=black, fill=white,thick]{\large$\lambda_R$}; 

\filldraw[color=black, dashed, fill=white, thick](0,0) rectangle (1,1);
\filldraw[color=black, fill=white, thick](1,0) rectangle (2,1);

\filldraw[color=black, fill=white, thick](6,0) rectangle (7,1);
\filldraw[color=black, dashed, fill=white, thick](7,0) rectangle (8,1);

\filldraw[color=black, fill=white, thick](3,3) rectangle (4,4);
\filldraw[color=black, fill=white, thick](4,3) rectangle (5,4);

\node at (2.5, 0.5) [draw,scale=1.5,color=red, fill=white,thick]{$A$}; 
\node at (5.5, 0.5) [draw,scale=1.5,color=red, fill=white,thick]{$A$};
\node at (2.5, 3.5) [draw,scale=1.5,color=red, fill=white,thick]{$A$}; 
\node at (5.5, 3.5) [draw,scale=1.5,color=red, fill=white,thick]{$A$};  

\end{tikzpicture}}
    \caption{A complete non-abelian gauge-invariant theory over $\psi$ and $A$\label{fig:2full}. Whenever a black right-moving (resp. left-moving) wire for $\psi$ crosses a red wire for $A$, then $A$ (resp. $A^{-1}$) gets applied upon $\psi$.}
\end{figure}
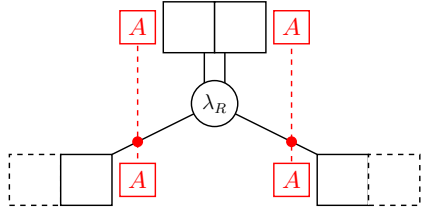

\begin{figure}[ht!]
    \centering
    \resizebox{0.8\textwidth}{!}{\begin{tikzpicture}

\newcommand{\framee}[2]{
    \begin{scope}[shift={(#1,#2)}]
        \draw[thick] (2,1) -- (4,2)
            (6,1) -- (4,2)
            (3.8,2) -- (3.8, 3)
            (4.2,2) -- (4.2, 3);
        \node at (4, 2) [draw,scale=3,circle,color=black, fill=white,thick]{};
    \end{scope}
} 
\framee{0}{0};
\framee{6}{0};
\framee{12}{0};
\draw (0,1) -- (-1,1.5) (20,1) -- (21,1.5);
\filldraw[color=black, fill=white, thick](0,0) rectangle (1,1);
\filldraw[color=black, fill=white, thick](1,0) rectangle (2,1);
\filldraw[color=black, fill=white, thick](6,0) rectangle (7,1);
\filldraw[color=black, fill=gray, thick](7,0) rectangle (8,1);
\filldraw[color=black, fill=black, thick](12,0) rectangle (13,1);
\filldraw[color=black, fill=white, thick](13,0) rectangle (14,1);
\filldraw[color=black, fill=white, thick](18,0) rectangle (19,1);
\filldraw[color=black, fill=white, thick](19,0) rectangle (20,1);

\framee{-3}{3};
\framee{3}{3};
\framee{9}{3};
\framee{15}{3};
\filldraw[color=black, fill=white, thick](3,3) rectangle (4,4);
\filldraw[color=black, fill=white, thick](4,3) rectangle (5,4);
\filldraw[color=black, fill=white, thick](9,3) rectangle (10,4);
\filldraw[color=black, fill=gray, thick](10,3) rectangle (11,4);
\filldraw[color=black, fill=white, thick](15,3) rectangle (16,4);
\filldraw[color=black, fill=white, thick](16,3) rectangle (17,4);

\filldraw[color=black, fill=white, thick](0,6) rectangle (1,7);
\filldraw[color=black, fill=white, thick](1,6) rectangle (2,7);
\filldraw[color=black, fill=white, thick](6,6) rectangle (7,7);
\filldraw[color=black, fill=white, thick](7,6) rectangle (8,7);
\filldraw[color=black, fill=white, thick](12,6) rectangle (13,7);
\filldraw[color=black, fill=gray, thick](13,6) rectangle (14,7);
\filldraw[color=black, fill=white, thick](18,6) rectangle (19,7);
\filldraw[color=black, fill=white, thick](19,6) rectangle (20,7);

\draw (1, -0.5) node {\Large $x-3$};
\draw (4, -0.5) node {\Large $x-2$};
\draw (7, -0.5) node {\Large $x-1$};
\draw (10, -0.5) node {\Large $x$};
\draw (13, -0.5) node {\Large $x+1$};
\draw (16, -0.5) node {\Large $x+2$};
\draw (19, -0.5) node {\Large $x+3$};

\draw (-1.5, 0.5) node {\Large $t$};
\draw (-1.5, 3.5) node {\Large $t+1$};
\draw (-1.5, 6.5) node {\Large $t+2$};

\newcommand{\frameg}[2]{
    \ifthenelse{#2 > 0}{\def\mycol{white}}{\def\mycol{red}};
    \draw[thick, dashed, red] (2.5+3*#1,0.5) -- (2.5+3*#1,6.5);
    \node at (2.5+3*#1, 0.5) [draw,scale=2,color=red, fill=\mycol,thick]{};
    \node at (2.5+3*#1, 3.5) [draw,scale=2,color=red, fill=\mycol,thick]{};
    \node at (2.5+3*#1, 6.5) [draw,scale=2,color=red, fill=\mycol,thick]{};
    \node at (2.5+3*#1, 1.25) [draw,scale=0.7,circle,color=red, fill=red,thick]{};
    \node at (2.5+3*#1, 4.25) [draw,scale=0.7,circle,color=red, fill=red,thick]{};
}

\frameg{0}{1};
\frameg{1}{1};
\frameg{2}{1};
\frameg{3}{0};
\frameg{4}{1};
\frameg{5}{1};

\end{tikzpicture}}
    \caption{The complete theory is richer than the initial theory\label{fig:2fullexample}. Here an empty circle for $A$ represent the identity while a full circle represents the permutation of white and black colours (leaving gray untouched). At position $x+1/2$, the input coming from $x+1$ is toggled from black to white.}
\end{figure}
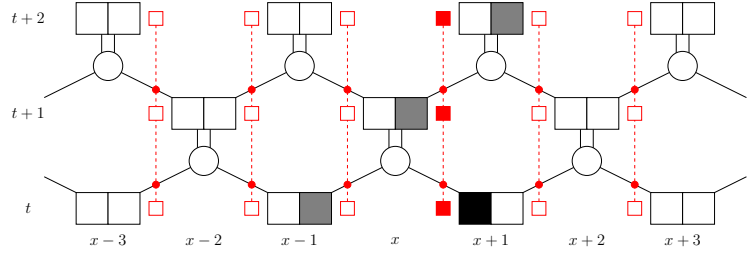

This fully non-abelian gauge-invariant cellular automaton was built through the simplest possible choices via the gauging procedure. Notice that it is not so trivial however : the introduction of the gauge field $A$, motivated by the will to restore gauge-invariant, ends up truly enriching the phenomenology of the theory.

Having developed a gauge-invariant theory means having manage to introduce\ldots a redundancy. Thus, there will be other gauge-invariant theories that are equivalent, up to that redudancy. Can we characterize those \textit{equivalent} theories?

\section{Equivalence and invariant sets\label{sec:equivalence}}


Given a set of gauge transformations $\Gamm$, multiple theories may lead to similar dynamics with respect to $\Gamm$: 
\begin{definition}[Equivalence of two theories] Let $T$ be a gauge-invariant theory with respect to $\Gamm$ and $Z$. $T$ is simulated by $T'$ if and only if for all configuration $c$ there exists $\gamm, \gamm' \in \Gamm$ such that $(\gamm' \circ T)(c) = (T'\circ \gamm)(c)$. They are equivalent if both simulate each other. We denote the equivalence as $T\equiv T'$. 
\end{definition}
Thus $T\equiv T'$ if and only if they give rise to the same dynamics up to a gauge transformation.

$T$ is gauge-invariant with respect to a specific $Z$. Adding some constraints on $Z$ and $\Gamma$, one may characterize the equivalence of two theories using different quantifiers and constraints which may be useful for some specific problems. More specifically, it will be easier to prove that two theories are equivalent using the characterization than the definition.
\begin{proposition}[Characterization of equivalence of theories]
Let $T$ be a gauge-invariant theory with respect to $\Gamm$ and $Z$. If $Z$ is reversible and $\Gamma$ is a group, then $T$ is simulated by $T'$ if and only if
\begin{enumerate}
    \item $\forall c, \exists \gamm \in \Gamm$ such that $T(c) = T'\circ \gamm (c)$.
    \item $\forall c, \forall \gamm \in \Gamm$, $\exists \gamm' \in \Gamm$ such that $\gamm' \circ T(c) = T'\circ \gamm (c)$.
\end{enumerate}
\end{proposition}
\begin{proof}
We shall prove the equivalence through three implications.
\begin{itemize}
    \item The fact that (3) implies (1) is immediate. 
    
    \item Suppose (1), then for $c$ a configuration, we have $\gamm, \gamm' \in \Gamm$ such that $(\gamm'\circ T)(c) = (T'\circ \gamm)(c)$. But since $\Gamm$ is a group, it implies that $T(c) = (\gamm'^{-1} \circ T' \circ \gamm)(c)$. And since $Z$ is reversible, we obtain $T(c) = (T'\circ Z^{-1} (\gamm'^{-1}) \circ \gamm) (c)$. However, $Z^{-1} (\gamm'^{-1}) \circ \gamm$ is an element of $\Gamm$ therefore we have proven that (1) implies (2).
    
    \item Suppose (2), let $c$ be a configuration and take $\gamm\in\Gamm$ such that $T(c)=(T'\circ\gamm)(c)$. For any $\gamm_1 \in \Gamm$ there exists $\gamm_3 \in \Gamm$ such that $\gamm = \gamm_3 \circ \gamm_1$. Therefore, $T(c) = (Z(\gamm_3) \circ T' \circ \gamm_1) (c)$ which is equivalent to $(Z(\gamm_3)^{-1} \circ T)(c) = (T\circ \gamm_1) (c)$. And writing $\gamm_2 = Z(\gamm_3)^{-1}$ which is in $\Gamm$, we conclude that (2) implies (3).
\end{itemize} \qed
\end{proof}

\paragraph{Invariant.} $\Gamm$ defines a set of transformations over the set of configurations. When one configuration can be transformed into another, the two are thought of as physically equivalent. For a gauge-invariant theory $T$, equivalent configurations with respect to $\Gamm$ lead to equivalent configurations after the evolution $T$. Therefore, one may think that such a theory would be better formulated to act over the set of equivalence classes of configurations instead, i.e. those sets that are left invariant under $\Gamm$. Formally, for $\Sigma$ the internal state space of $T$, for $\psi \in \Sigma$, let $I_\psi = \{ \gamma(\psi) \; | \; \gamma \in \Gamma \}$. If $\Gamma$ is a group, which was the case in our running example, then for all $\psi$ and $\psi'$ :
$$ \exists \gamma \in \Gamma,\,\psi'=\gamm(\psi)\ \Longleftrightarrow\ I_\psi = I_{\psi'}.$$
Then $T$ is indeed equivalent to a theory $T'$ having these invariant sets at its internal state space---or rather the canonical representant elements of these.\\
However, for an inhomogeneous gauge-invariant theory $R_A$ one needs to be more careful. that is  this is true only if the invariant sets are built from $(\psi, A)$ and not just $\psi$. Indeed, an invariant set built only over $\psi$ and not considering $A$ would be like disregarding the gauge transformation of $A$ and therefore, breaking the inhomogeneous gauge-invariance. 

Moreover, to be more subtle, it is not enough to consider the invariant sets for $\psi$ and $A$ separately. Indeed the invariant set for $(\psi, A)$ is generally not the cartesian product of the invariant set of $\psi$ with that of $A$, because the gauge-transformation acts on both $\psi$ and $A$ synchronously. A simple example is given Fig-\ref{fig:invariants}, which works already when we restrict ourselves to $\Sigma=\{0,1\}^2$ for simplicity, i.e. back in the abelian case. It starts with both sides having the exact same $\psi$ but two different gauge fields related by a gauge-transformation. Here the gauge-transformation applied on $A$ is the identity everywhere except at position $x$, for which $\gamma^l=\gamma^r$ is the permutation of $0$ and $1$. This means both sub-figures have the same invariant sets for $A$ because they are linked through a gauge-transformation, and idem for $\psi$. After a time step however, the invariant sets for $\psi$ are not identical on both sides, because we cannot consider $\psi$ and $A$ separately when looking at the invariant sets. Again this problem does not appear when considering the invariant set for the couple $(\psi, A)$ directly.

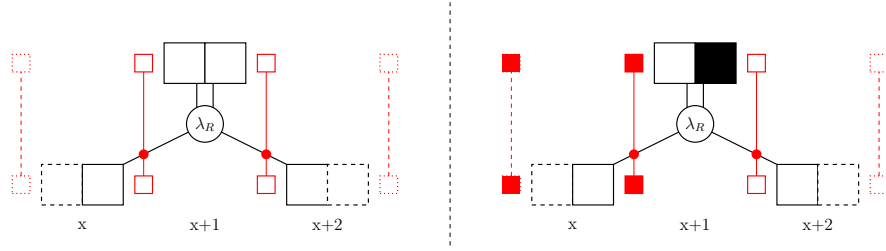
\begin{figure}[ht!]
    \centering
    \resizebox{0.95\textwidth}{!}{\begin{tikzpicture}

\draw[dashed] (0,-1) -- (0,5);

\draw[color=red] (4.5, 0.5) -- (4.5,3.5)
                (7.5, 0.5) -- (7.5,3.5);
\draw[color=red,dashed] (1.5, 0.5) -- (1.5,3.5)
                (10.5, 0.5) -- (10.5,3.5);
\draw[thick] (4,1) -- (6,2)
    (8,1) -- (6,2)
    (5.8,2) -- (5.8, 3)
    (6.2,2) -- (6.2, 3);
    
\node at (6, 2) [draw,scale=1,circle,color=black, fill=white,thick]{\large$\lambda_{R}$}; 
\filldraw[color=black, fill=white, thick, dashed](2,0) rectangle (3,1);
\filldraw[color=black, fill=white, thick](3,0) rectangle (4,1);
\filldraw[color=black, fill=white, thick](8,0) rectangle (9,1);
\filldraw[color=black, fill=white, thick, dashed](9,0) rectangle (10,1);
\node at (4.5, 0.5) [draw,scale=2,color=red, fill=red,thick]{}; 
\node at (4.5, 3.5) [draw,scale=2,color=red, fill=red,thick]{}; 
\node at (7.5, 0.5) [draw,scale=2,color=red, fill=white,thick]{}; 
\node at (7.5, 3.5) [draw,scale=2,color=red, fill=white,thick]{}; 

\node at (4.5, 1.25) [draw,scale=0.7,circle,color=red, fill=red,thick]{}; 
\node at (7.5, 1.25) [draw,scale=0.7,circle,color=red, fill=red,thick]{}; 

\node at (1.5, 0.5) [draw,scale=2,color=red, fill=red,thick, dotted]{}; 
\node at (1.5, 3.5) [draw,scale=2,color=red, fill=red,thick, dotted]{}; 
\node at (10.5, 0.5) [draw,scale=2,color=red, fill=white,thick, dotted]{}; 
\node at (10.5, 3.5) [draw,scale=2,color=red, fill=white,thick, dotted]{}; 
\filldraw[color=black, fill=white, thick](5,3) rectangle (6,4);
\filldraw[color=black, fill=black, thick](6,3) rectangle (7,4);

\draw (3,-0.5) node {\large x};
\draw (6,-0.5) node {\large x+1};
\draw (9,-0.5) node {\large x+2};

\begin{scope}[shift={(-12,0)}]
    \draw[color=red] (4.5, 0.5) -- (4.5,3.5)
                    (7.5, 0.5) -- (7.5,3.5);
    \draw[color=red,dashed] (1.5, 0.5) -- (1.5,3.5)
                    (10.5, 0.5) -- (10.5,3.5);
    \draw[thick] (4,1) -- (6,2)
        (8,1) -- (6,2)
        (5.8,2) -- (5.8, 3)
        (6.2,2) -- (6.2, 3);
        
    \node at (6, 2) [draw,scale=1,circle,color=black, fill=white,thick]{\large$\lambda_{R}$}; 
    \filldraw[color=black, fill=white, thick, dashed](2,0) rectangle (3,1);
    \filldraw[color=black, fill=white, thick](3,0) rectangle (4,1);
    \filldraw[color=black, fill=white, thick](8,0) rectangle (9,1);
    \filldraw[color=black, fill=white, thick, dashed](9,0) rectangle (10,1);
    \node at (4.5, 0.5) [draw,scale=2,color=red, fill=white,thick]{}; 
    \node at (4.5, 3.5) [draw,scale=2,color=red, fill=white,thick]{}; 
    \node at (7.5, 0.5) [draw,scale=2,color=red, fill=white,thick]{}; 
    \node at (7.5, 3.5) [draw,scale=2,color=red, fill=white,thick]{};
    
    \node at (4.5, 1.25) [draw,scale=0.7,circle,color=red, fill=red,thick]{}; 
    \node at (7.5, 1.25) [draw,scale=0.7,circle,color=red, fill=red,thick]{}; 
    
    \node at (1.5, 0.5) [draw,scale=2,color=red, fill=white,thick, dotted]{}; 
    \node at (1.5, 3.5) [draw,scale=2,color=red, fill=white,thick, dotted]{}; 
    \node at (10.5, 0.5) [draw,scale=2,color=red, fill=white,thick, dotted]{}; 
    \node at (10.5, 3.5) [draw,scale=2,color=red, fill=white,thick, dotted]{}; 
    \filldraw[color=black, fill=white, thick](5,3) rectangle (6,4);
    \filldraw[color=black, fill=white, thick](6,3) rectangle (7,4);
    
    \draw (3,-0.5) node {\large x};
    \draw (6,-0.5) node {\large x+1};
    \draw (9,-0.5) node {\large x+2};
\end{scope}

\end{tikzpicture}}
    \caption{Both sub-figures initially have the same invariant sets for $\psi$ and $A$ respectively. After a time step, this is not true for $\psi$: they do not share an invariant set. This figure shows that it is not enough to consider the invariant sets for $\psi$ and $A$ separately }
    \label{fig:invariants}
\end{figure}

\section{Conclusion\label{sec:conclusion}}

\paragraph{Summary.} In this paper, we reformulated and generalized the theory of gauge-invariance in CA \cite{arrighi2018gauge}, to cater for non-abelian symmetry groups.
The gauging procedure was then made explicit and developed through an example : starting from a non-gauge-invariant theory and a set of gauge transformations, we introduced an external gauge field upon which the theory was made dependant, and extended the gauge transformation to this field, so as to obtain gauge-invariance. Finally the gauge field was `internalized' by providing a theory for its dynamics, yielding a complete gauge-invariant theory. Now, gauge-invariant theories are redundant almost by definition, and thus several theories may be equivalent to one another up to this redundancy. Equivalence between theories was formalized and characterized. Configuration that are related by a gauge-transformations were gathered into invariant sets, called invariant sets, and CA over these invariant sets were discussed.

\paragraph{Perspectives and related works.} Since gauge-invariance comes from Physics, the first extension of this model would be a non-abelian gauge-invariant Quantum CA (QCA). An abelian gauge-invariant QCA was already provided in \cite{arrighi2019quantum}, whereas a non-abelian gauge-invariance has been studied in the one-particle sector of QCA, namely quantum walks \cite{arnault2016quantum}: this extension is rather promising. \\ 
In the field of quantum computation, gauge-invariance is already mentioned for quantum error correction codes \cite{kitaev2003fault,nayak2008non} which can be understood through the redundancy inherent to gauge-invariant theories. The study of gauge-invariance in CA ought to be related, therefore, to questions of error correction for spatially-distributed computation models \cite{Toom,TVSMKP}.\\
Finally, gauge-invariance brings another symmetry to field CA, which may be interesting to study for itself, e.g. along the same methods used for color-blind CA \cite{salo2013}, where all cells get transformed by the same group element.

\bibliographystyle{splncs04}
\bibliography{biblio}

\end{document}